\documentclass[11pt]{article}

\usepackage{amsmath}
\usepackage{amssymb}
\usepackage{amsfonts,amsthm}
\usepackage{thmtools,thm-restate}
\usepackage{tabularx}

\newtheorem{theorem}{Theorem}[section]
\newtheorem{lemma}[theorem]{Lemma}
\newtheorem{corollary}[theorem]{Corollary}

\newtheorem{proposition}[theorem]{Proposition}
\newtheorem{observation}[theorem]{Observation}

\theoremstyle{definition}
\newtheorem{definition}[theorem]{Definition}

\usepackage{tikz}
\usepackage[a4paper,bindingoffset=0.2in,
            top=1.1in, bottom=1.1in, left=1.0in, right=1.0in,
            footskip=.5in]{geometry}

\newcommand{\ignore}[1]{}

\title{Cluster Deletion on Interval Graphs and Split Related Graphs}

\author{
Athanasios L. Konstantinidis\thanks{Department of Mathematics, University of Ioannina, Greece. E-mail: \texttt{skonstan@cc.uoi.gr}.
This research has been financially supported by the State Scholarships Foundation (IKY).}
\and
Charis Papadopoulos\thanks{Department of Mathematics, University of Ioannina, Greece. E-mail:  \texttt{charis@cs.uoi.gr}}
}

\date{}

\begin{document}

\maketitle

\begin{abstract}
In the {\sc Cluster Deletion} problem the goal is to remove the minimum number of edges of a given graph, such that every connected component of the resulting graph constitutes a clique.
It is known that the decision version of {\sc Cluster Deletion} is NP-complete on ($P_5$-free) chordal graphs, whereas {\sc Cluster Deletion} is solved in polynomial time on split graphs.
However, the existence of a polynomial-time algorithm of {\sc Cluster Deletion} on interval graphs, a proper subclass of chordal graphs, remained a well-known open problem.
Our main contribution is that we settle this problem in the affirmative, by providing a polynomial-time algorithm for {\sc Cluster Deletion} on interval graphs.
Moreover, despite the simple formulation of the algorithm on split graphs, we show that {\sc Cluster Deletion} remains NP-complete on a natural and slight generalization of split graphs that constitutes a proper subclass of $P_5$-free chordal graphs.
To complement our results, we provide two polynomial-time algorithms for {\sc Cluster Deletion} on subclasses of such generalizations of split graphs.
\end{abstract}

\section{Introduction}
In graph theoretic notions, {\it clustering} is the task of partitioning the vertices of the graph into subsets, called clusters, in such a way that there should be
many edges within each cluster and relatively few edges between the clusters.
In many applications, the clusters are restricted to induce cliques, as the represented data of each edge corresponds to a similarity value between two objects \cite{HJ97,Hartigan}.
Under the term cluster graph. which refers to a disjoint union of cliques, one may find a variety of applications that have been extensively studied \cite{BBC04,CGW03,Schaeffer07}.
Here we consider the {\sc Cluster Deletion} problem which asks for a minimum number of edge deletions from an input graph, so that the resulting graph is a disjoint union of cliques.
In the decision version of the problem, we are also given an integer $k$ and we want to decide whether at most $k$ edge deletions are enough to produce a cluster graph.

Although {\sc Cluster Deletion} is NP-hard on general graphs \cite{CluNP04}, settling its complexity status restricted on graph classes has attracted several researchers.
Regarding the maximum degree of a graph, Komusiewicz and Uhlmann \cite{CD-lowdegree} have shown an interesting complexity dichotomy result:
{\sc Cluster Deletion} remains NP-hard on $C_4$-free graphs with maximum degree four, whereas it can be solved in polynomial time for graphs having maximum degree at most three.
Quite recently, Golovach et al. have shown that it remains NP-hard on planar graphs \cite{GolovachHKLP18}.
For graph classes characterized by forbidden induced subgraphs, Gao et al. \cite{CD-cographs} showed that {\sc Cluster Deletion} is NP-hard on $(C_5, P_5, \textrm{bull}, \textrm{fork}, \textrm{co-gem}, \textrm{4-pan}, \textrm{co-4-pan})$-free graphs and on $(2K_2,3K_1)$-free graphs.
Regarding $H$-free graphs, Gr{\"{u}}ttemeier et al. \cite{GK18}, showed a complexity dichotomy result for any graph $H$ consisting of at most four vertices.
In particular, for any graph $H$ on four vertices such that $H \notin \{P_4, \textrm{paw}\}$, {\sc Cluster Deletion} is NP-hard on $H$-free graphs,
whereas it can be solved in polynomial time on $P_4$- or $\textrm{paw}$-free graphs \cite{GK18}.
Interestingly, {\sc Cluster Deletion} remains NP-hard on $P_5$-free chordal graphs \cite{BDM15}.

On the positive side, {\sc Cluster Deletion} have been shown to be solved in polynomial time on cographs \cite{CD-cographs},
proper interval graphs \cite{BDM15},
split graphs \cite{BDM15},
and $P_4$-reducible graphs \cite{BonomoDNV15}.
More precisely, iteratively picking maximum cliques defines a clustering on the graph which actually gives an optimal solution on cographs (i.e., $P_4$-free graphs),
as shown by Gao et al. in \cite{CD-cographs}.
In fact, the greedy approach of selecting a maximum clique provides a $2$-approximation algorithm, though not necessarily in polynomial-time \cite{DessmarkJLLP07}.
As the problem is already NP-hard on chordal graphs \cite{BDM15}, it is natural to consider subclasses of chordal graphs such as interval graphs and split graphs.
Although for split graphs there is a simple polynomial-time algorithm, restricted to interval graphs only the complexity of a proper subclass, proper interval graphs, was determined by giving a solution that runs in polynomial-time \cite{BDM15}.
Settling the complexity of {\sc Cluster Deletion} on interval graphs, was left open \cite{BDM15,BonomoDNV15,CD-cographs}.

For proper interval graphs, Bonomo et al. \cite{BDM15} characterized their optimal solution by consecutiveness of each cluster
with respect to their natural ordering of the vertices.
Based on this fact, a dynamic programming approach led to a polynomial-time algorithm.
It is not difficult to see that such a consecutiveness does not hold on interval graphs, as potential clusters might require to break in the corresponding vertex ordering.
Here we characterize an optimal solution of interval graphs whenever a cluster is required to break.
In particular, we take advantage of their consecutive arrangement of maximal cliques and describe subproblems of maximal cliques containing the last vertex.
One of our key observations is that the candidate clusters containing the last vertex can be enumerated in polynomial time given two vertex orderings of the graph. 
We further show that each such candidate cluster separates the graph in a recursive way with respect to optimal subsolutions, 
that enables to define our dynamic programming table to keep track about partial solutions.
Thus, our algorithm for interval graphs suggests to consider a particular consecutiveness of a solution and apply a dynamic programming approach defined by two vertex orderings. 

Furthermore, we complement the previously-known NP-harness of {\sc Cluster Deletion} on $P_5$-free chordal graphs, by providing a proper subclass of such graphs for which we prove that the problem remains NP-hard.
This result is inspired and motivated by the very simple characterization of an optimal solution on split graphs: 
either a maximal clique constitutes the only non-edgeless cluster, or there are exactly two non-edgeless clusters whenever there is a vertex of the independent set that is adjacent
to all the vertices of the clique except one \cite{BDM15}. 
Due to the fact that true twins belong to the same cluster in an optimal solution, it is natural to consider true twins at the independent set, 
as they are expected not to influence the solution characterization.  
Surprisingly, we show that {\sc Cluster Deletion} remains NP-complete even on such a slight generalization of split graphs. 
We then study two different classes of such generalization of split graphs that can be viewed as the parallel of split graphs that admit disjoint clique-neighborhood and nested clique-neighborhood.
For {\sc Cluster Deletion} we provide polynomial-time algorithms on both classes of graphs.


\section{Preliminaries}
All graphs considered here are simple and undirected.
A graph is denoted by $G=(V,E)$ with vertex set $V$ and edge set $E$. We use
the convention that $n=|V|$ and $m=|E|$.
The {\it neighborhood} of
a vertex~$v$ of $G$ is $N(v)=\{x \mid vx \in E\}$ and the
{\it closed neighborhood} of $v$ is $N[v] = N(v) \cup \{v\}$.
For $S \subseteq V$, $N(S)=\bigcup_{v \in S} N(v) \setminus S$ and $N[S] = N(S) \cup S$.
A graph~$H$ is a {\it subgraph} of $G$ if $V(H)\subseteq V(G)$
and $E(H)\subseteq E(G)$. For $X\subseteq V(G)$, the subgraph
of $G$ {\it induced} by $X$, $G[X]$, has vertex set~$X$, and
for each vertex pair~$u, v$ from $X$, $uv$ is an edge of $G[X]$
if and only if $u\not= v$ and $uv$ is an edge of $G$.
For $R\subseteq E(G)$, $G\setminus R$ denotes the
graph~$(V(G), E(G)\setminus R)$, that is a subgraph of $G$ and
for $S \subseteq V(G)$, $G - S$ denotes the
graph~$G[V(G)-S]$, that is an induced subgraph of $G$.
For two set of vertices $A$ and $B$, we write $E(A,B)$ to denote the edges that have one endpoint in $A$ and one endpoint in $B$.
Two adjacent vertices $u$ and $v$ are called {\it true twins} if $N[u] = N[v]$,
whereas two non-adjacent vertices $x$ and $y$ are called {\it false twins} if $N(u)=N(v)$.

A {\it clique} of $G$ is a set of pairwise adjacent vertices
of $G$, and a {\it maximal clique} of $G$ is a clique of $G$
that is not properly contained in any clique of $G$. An
{\it independent set} of $G$ is a set of pairwise non-adjacent
vertices of $G$.
For $k\geq 2$, the chordless path on $k$ vertices is denoted by $P_k$ and the
chordless cycle on $k$ vertices is denoted by $C_k$.
For an induced path $P_k$, the vertices of degree one are called endvertices.
A vertex $v$ is {\it universal} in $G$ if $N[v] = V(G)$ and $v$ is {\it isolated} if $N(v) = \emptyset$.
A graph is {\it connected} if there is a path between any
pair of vertices.
A {\it connected component} of $G$ is a maximal connected subgraph of $G$.
For a set of finite graphs $\mathcal{H}$, we say that a graph $G$ is $\mathcal{H}$-free if $G$ does not contain an induced subgraph isomorphic to any of the graphs of $\mathcal{H}$.

\medskip
The problem of {\sc Cluster Deletion} is formally defined as follows: given a graph $G=(V,E)$, the goal is to compute the minimum set $F \subseteq E(G)$ of edges
such that every connected component of $G - F$ is a clique.
A {\it cluster graph} is a $P_3$-free graph, or equivalently, any of its connected components is a clique.
Thus, the task of {\sc Cluster Deletion} is to turn the input graph $G$ into a cluster graph by deleting the minimum number of edges.
Let $S = C_1, \ldots, C_k$ be a solution of {\sc Cluster Deletion} such that $G[C_i]$ is a clique.
In such terms, the problem can be viewed as a vertex partition problem into $C_1, \ldots, C_k$.
Each $C_i$ is simple called {\it cluster}.
Edgeless clusters, i.e., clusters containing exactly one vertex, are called {\it trivial clusters}.
The edges of $G$ are partitioned into {\it internal} and {\it external} edges: an internal edge $uv$ has both its endpoints $u,v \in C_i$ in the same cluster $C_i$, whereas an external edge $uv$ has its endpoints in different clusters $u \in C_i$ and $v \in C_j$, for $i\neq j$.
Then, the goal of {\sc Cluster Deletion} is to minimize the number of external edges which is equivalent to maximize the number of internal edges.
We write $S(G)$ to denote an optimal solution for {\sc Cluster Deletion} of the graph $G$, that is, a cluster subgraph of $G$ having the maximum number of edges.
Given a solution $S(G)$, the number of edges incident only to the same cluster, that is the number of internal edges, is denoted by $|S(G)|$.

For a clique $C$, we say that a vertex $x$ is {\it $C$-compatible} if $C\setminus\{x\} \subseteq N(x)$. 
We start with few preliminary observations regarding twin vertices.
Notice that for true twins $x$ and $y$, if $x$ belongs to any cluster $C$ then $y$ is $C$-compatible.

\begin{lemma}[\cite{BDM15}]\label{lem:truetwin}
Let $x$ and $y$ be true twins in $G$. Then, in any optimal solution $x$ and $y$ belong to the same cluster.
\end{lemma}

The above lemma shows that we can contract true twins and look for a solution on a vertex-weighted graph that does not contain true twins.
Even though false twins cannot be grouped into the same cluster as they are non-adjacent, we can actually disregard one of the false twins whenever their neighborhood forms a clique.
\begin{lemma}\label{lem:falsetwin}
Let $x$ and $y$ be false twins in $G$ such that $N(x)=N(y)$ is a clique.
Then, there is an optimal solution such that $x$ constitutes a trivial cluster.
\end{lemma}
\begin{proof}
Let $C_x$ and $C_y$ be the clusters of $x$ and $y$, respectively, in an optimal solution such that $|C_x| \geq 2$ and $|C_y| \geq 2$.
We construct another solution by replacing both clusters by $C_x \cup C_y \setminus \{y\}$ and $\{y\}$, respectively.
To see that this indeed a solution, first observe that $x$ is adjacent to all the vertices of $C_y \setminus \{y\}$ because $N(x)=N(y)$, and $C_x \cup C_y \setminus \{y\} \subseteq N[x]$ forms a clique by the assumption.
Moreover, since $|C_x| \geq 2$ and $|C_y| \geq 2$, we know that $|C_x|+|C_y| \leq |C_x||C_y|$, implying that the number of internal edges in the constructed solution is at least as the number of internal edges of the optimal solution.
\end{proof}

Moreover, we prove the following generalization of Lemma~\ref{lem:truetwin}.
\begin{lemma}\label{lem:compatible}
Let $C$ and $C'$ be two clusters of an optimal solution and let $x\in C$ and $y \in C'$.
If $y$ is $C$-compatible then $x$ is not $C'$-compatible.
\end{lemma}
\begin{proof}
Let $S$ be an optimal solution such that $C, C' \in S$.
Assume for contradiction that $x$ is $C'$-compatible.
We show that $S$ is not optimal.
Since $y$ is $C$-compatible, we can move $y$ to $C$ and obtain a solution $S_y$ that contains the clusters $C \cup \{y\}$ and $C' \setminus \{y\}$.
Similarly, we construct a solution $S_x$ from $S$, by moving $x$ to $C'$ so that $C\setminus\{x\}, C' \cup \{x\} \in S_x$.
Notice that the $S_x$ forms a clustering, since $x$ is $C'$-compatible.
We distinguish between the following cases, according to the values $|C|$ and $|C'|$.
\begin{itemize}
\item If $|C| \geq |C'|$ then $|S_y|>|S|$, because ${{|C|+1} \choose 2} + {{|C'|-1} \choose 2} > {{|C|} \choose 2} + {{|C'|} \choose 2}$.
\item If $|C| < |C'|$ then $|S_x|>|S|$, because ${{|C|-1} \choose 2} + {{|C'|+1} \choose 2} > {{|C|} \choose 2} + {{|C'|} \choose 2}$.
\end{itemize}
In both cases we reach a contradiction to the optimality of $S$.
Therefore, $x$ is not $C'$-compatible.
\end{proof}

\begin{corollary}\label{cor:inclusion}
Let $C$ be a cluster of an optimal solution and let $x\in C$.
If there is a vertex $y$ that is $C$-compatible and $N[y] \subseteq N[x]$, then $y$ belongs to $C$.
\end{corollary}
\begin{proof}
Assume for contradiction that $y$ belongs to a cluster $C'$ different than $C$.
Then, observe that $x$ is $C'$-compatible. Indeed, for any vertex $u$ of $C'$, we know $xu \in E(G)$, since $u$ is adjacent to $y$ and $N[y] \subseteq N[x]$.
Thus, by Lemma~\ref{lem:compatible} we reach a contradiction, so that $y \in C$.
\end{proof}

%
%


\section{Polynomial-time algorithm on interval graphs}
Here we present a polynomial-time algorithm for the {\sc Cluster Deletion} problem
on interval graphs.
A graph is an \emph{interval graph} if there is a bijection between its vertices and a
family of closed intervals of the real line such that two vertices are adjacent if and
only if the two corresponding intervals intersect.
Such a bijection is called an \emph{interval representation} of the graph, denoted by $\mathcal{I}$.
We identify the intervals of the given representation with the vertices of the graph, interchanging these notions appropriately.
Whether a given graph is an interval graph can be decided in linear time and if so, an interval representation can be generated in linear time~\cite{FG65}.
Notice that every induced subgraph of an interval graph is an interval graph.


Let $G$ be an interval graph. Instead of working with the interval representation of $G$, we consider its sequence of maximal cliques.
It is known that a graph $G$ with $p$ maximal cliques is an interval graph if and only if
there is an ordering $K_1, \ldots, K_{p}$ of the maximal cliques of $G$, such that for each vertex $v$ of $G$,
the maximal cliques containing $v$ appear consecutively in the ordering (see e.g., \cite{BraLeSpi99}).
A path $\mathcal{P} = K_1 \cdots K_{p}$ following such an ordering is called a \emph{clique path} of $G$.
Notice that a clique path is not necessarily unique for an interval graph.
Also note that an interval graph with $n$ vertices contains at most $n$ maximal cliques.
By definition, for every vertex $v$ of $G$, the maximal cliques containing $v$ form a connected subpath in $\mathcal{P}$.

Given a vertex $v$, we denote by $K_{a(v)}, \ldots, K_{b(v)}$ the maximal cliques containing $v$ with respect to $\mathcal{P}$, where $K_{a(v)}$ and $K_{b(v)}$
are the {\it first} ({\it leftmost}) and {\it last} ({\it rightmost}) maximal cliques containing $v$. Notice that $a(v) \leq b(v)$ holds.
Moreover, for every edge of $G$ there is a maximal clique $K_i$ of $\mathcal{P}$ that contains both endpoints of the edge.
Thus, two vertices $u$ and $v$ are adjacent if and only if $a(v) \leq a(u) \leq b(v)$ or $a(v) \leq b(u) \leq b(v)$.

\newcommand{\amin}{a\text{-}\min}
\newcommand{\amax}{a\text{-}\max}
\newcommand{\bmin}{b\text{-}\min}
\newcommand{\bmax}{b\text{-}\max}

For a set of vertices $U \subseteq V$, we write $\amin U$ and $\amax U$ to denote the minimum and maximum value, respectively, among all $a(u)$ with $u\in U$.
Similarly, $\bmin U$ and $\bmax U$ correspond to the minimum and maximum value, respectively, with respect to $b(u)$.


With respect to the {\sc Cluster Deletion} problem, observe that for any cluster $C$ of a solution, we know that $C \subseteq K_i$ where $K_i \in \mathcal{P}$, as $C$ forms a clique.
A vertex $y$ is called {\it guarded} by two vertices $x$ and $z$ if
$$
\min\{a(x),a(z)\} \leq a(y) \text{ and } b(y) \leq \max\{b(x),b(z)\}.
$$
For a clique $C$, observe that $y$ is $C$-compatible if and only if there exists a maximal clique $K_i$ such that $C\subseteq K_i$ with $a(y) \leq i \leq b(y)$. 

\begin{lemma}\label{lem:inclusion}
Let $x,y,z$ be three vertices of $G$ such that $y$ is guarded by $x$ and $z$.
If $x$ and $z$ belong to the same cluster $C$ of an optimal solution and $y$ is $C$-compatible then $y \in C$.
\end{lemma}
\begin{proof}
To ease the presentation, for three non-negative numbers $i,j,k$ we write $i \in [j,k]$ if $j \leq i \leq k$ holds.
Without loss of generality, assume that $a(y) \in [a(x), a(z)]$.
Assume for contradiction that $y$ belongs to another cluster $C'$.
We apply Lemma~\ref{lem:compatible} to either $x$ and $y$ or $z$ and $y$.
To do so, we need to show that $x$ is $C'$-compatible or $z$ is $C'$-compatible, as $y$ is already $C$-compatible.
Since $C'$ is a cluster that contains $y$, there is a maximal clique $K_i$ such that $C' \subseteq K_i$ with $i \in [a(y),b(y)]$.

We show that $i \in [a(x),b(x)]$ or $i \in [a(z),b(z)]$.
If $i \notin [a(x),b(x)]$ then $b(x) <i \leq b(y)$, because $a(x) \leq a(y) \leq i$.
As $y$ is guarded by $x$ and $z$, we know that $i \leq b(y) \leq b(z)$.
Now observe that if $i<a(z)$ then $b(x) < a(z)$, implying that $x$ and $z$ are non-adjacent, reaching a contradiction to the fact that $x,z \in C$.
Thus, $a(z) \leq i \leq b(z)$ which shows that $i \in [a(z),b(z)]$. This means that $i \in [a(x),b(x)]$ or $i \in [a(z),b(z)]$.

Hence, $x$ or $z$ belong to the maximal clique $K_i$ for which $C' \subseteq K_i$.
Therefore, at least one of $x$ or $z$ is $C'$-compatible and by Lemma~\ref{lem:compatible} we conclude that $y \in C$.
\end{proof}

Let $v_1, \ldots, v_n$ be an ordering of the vertices such that $b(v_1) \leq \cdots \leq b(v_n)$.
For every $v_i,v_j$ with $b(v_i) \leq b(v_j)$, we define the following set of vertices:
$$
V_{i,j} = \left\{v \in V(G): \min\{a(v_i),a(v_j)\}\leq a(v) \text{ and } b(v)\leq b(v_j) \right\}.
$$
That is, $V_{i,j}$ contains all vertices that are guarded by $v_i$ and $v_j$.
We write $a(i,j)$ to denote the value of $\min\{a(v_i),a(v_j)\}$ and we simple write $K_{a(j)}$ and $K_{b(j)}$ instead of $K_{a(v_j)}$ and $K_{b(v_j)}$.
Notice that for a neighbor $u$ of $v_j$ with $u \in V_{i,j}$, we have either $a(v_j)\leq a(u)$ or $a(v_i) \leq a(u) \leq a(v_j)$.
This means that all neighbors of $v_j$ that are totally included (i.e., all vertices $u$ such that $a(v_j) \leq a(u) \leq b(u) \leq b(v_j)$) belong to $V_{i,j}$ for any $v_i$ with $b(v_i) \leq b(v_j)$.
To distinguish such neighbors of $v_j$, we define the following sets:
\begin{itemize}
\item $U(j)$ contains the neighbors $u \in V_{i,j}$ of $v_j$ such that $a(u) < a(v_j) \leq b(u) \leq b(v_j)$ (neighbors of $v_j$ in $V_{i,j}$ that partially overlap $v_j$).
\item $M(j)$ contains the neighbors $w \in V_{i,j}$ of $v_j$ such that $a(v_j) \leq a(w) \leq b(w) \leq b(v_j)$ (neighbors of $v_j$ that are totally included within $v_j$).
\end{itemize}

In the forthcoming arguments, we restrict ourselves to the graph induced by $V_{i,j}$.
It is clear that the first maximal clique that contains a vertex of $V_{i,j}$ is $K_{a(i,j)}$, whereas the last maximal clique is $K_{b(j)}$.

We now explain the necessary sets that our dynamic programming algorithm uses in order to compute an optimal solution of $G$.
For two vertices $v_i,v_j$ with $b(v_i) \leq b(v_j)$, we define the following:
\begin{itemize}
\item $A_{i,j}$ is the value of an optimal solution for {\sc Cluster Deletion} of the graph $G[V_{i,j}]$.
\end{itemize}
To ease the notation, when we say a cluster of $A_{i,j}$ we mean a cluster of an optimal solution of $G[V_{i,j}]$.
Notice that $A_{1,n}$ is the desired value for the whole graph $G$, since $V_{1,n} = V(G)$.

Our task is to construct the values for $A_{i,j}$ by taking into account all possible clusters that contain $v_j$. 
To do so, we show that (i) the number of clusters containing $v_j$ in $A_{i,j}$ is polynomial and (ii) 
each such candidate cluster containing $v_j$ separates the graph in a recursive way with respect to optimal subsolutions.

Observe that if $v_iv_j \in E(G)$ then $v_i\in U(j)$ if and only if $a(v_i) < a(v_j)$, whereas $v_i \in M(j)$ if and only if $a(v_j) \leq a(v_i)$; in the latter case, it is not difficult to see that $V_{i,j} = M(j) \cup \{v_j\}$, according to the definition of $V_{i,j}$.
Thus, whenever $v_i \in M(j)$ holds, we have $V_{i,j}=V_{j,j}$.
The candidates of a cluster of $A_{i,j}$ containing $v_j$ lie among $U(j)$ and $M(j)$.
Let us show with the next two lemmas that we can restrict ourselves into a polynomial number of such candidates.
To avoid repeating ourselves, in the forthcoming statements we let $v_i,v_j$ be two vertices with $b(v_i) \leq b(v_j)$.


\begin{lemma}\label{lem:lower}
Let $C$ be a cluster of $A_{i,j}$ containing $v_j$. If there is a vertex $w \in M(j)$ such that $w\in C$ then
there is a maximal clique $K_{t}$ with $a(v_j) \leq t \leq b(v_j)$ such that $K_{t} \cap M(j) \subseteq C$ and $C \cap M(j) \subseteq K_{t}$.
\end{lemma}
\begin{proof}
Since $v_j,w \in C$, we know that there is a maximal clique $K_{t}$ for which $C \subseteq K_{t}$ with $a(v_j) \leq a(w) \leq t \leq \min\{b(v_j),b(w)\}$.
We show that all other vertices of $K_{t} \cap M(j)$ are guarded by $v_j$ and $w$.
Notice that for every vertex $y \in M(j)$ we already know that $a(v_j) \leq a(y)$ and $b(y) \leq b(v_j)$. 
Thus, for every vertex $y \in M(j)$ we have $a(v_j)=\min\{a(v_j),a(w)\} \leq a(y)$ and $b(y) \leq \max\{b(v_j),b(w)\}$.
This means that all vertices of $K_{t} \cap M(j) \setminus\{w\}$ are guarded by $v_j$ and $w$.
Moreover, since $C \subseteq K_{t}$, we know that all vertices of $K_{t} \cap M(j)$ are $C$-compatible.
Therefore, we apply Lemma~\ref{lem:inclusion} to every vertex of $K_{t} \cap M(j)$, showing that $K_{t} \cap M(j) \subseteq C$.
Furthermore, there is no vertex of $M(j) \setminus K_{t}$ that belongs to $C$, because $C \subseteq K_{t}$.
\end{proof}

By Lemma~\ref{lem:lower}, we know that we have to pick the entire set $K_{t} \cap M(j)$ for constructing candidates to form a cluster that contains $v_j$ and some vertices of $M(j)$.
As there are at most $n$ choices for $K_{t}$, we get a polynomial number of such candidate sets.
We next show that we can construct polynomial number of candidate sets that contain $v_j$ and vertices of $U(j)$.
For doing so, we consider the vertices of $U(j)$ increasingly ordered with respect to their first maximal clique. 
More precisely, let $U(j)_{\leq a} = (u_1, \ldots, u_{|U(j)|})$ be an increasingly order of the vertices of $U(j)$ such that $a(u_1) \leq \cdots \leq a(u_{|U(j)|})$. 
The right part of Figure~\ref{fig:CD_UM} illustrates the corresponding case. 

\begin{lemma}\label{lem:upper}
Let $C$ be a cluster of $A_{i,j}$ containing $v_j$ and let $u_{q} \in U(j)_{\leq a}$.
If $u_{q} \in C$ then every vertex of $\{u_{q+1}, \ldots, u_{|U(j)|}\}$ that is $C$-compatible belongs to $C$.
\end{lemma}
\begin{proof}
Let $u$ be a vertex of $\{u_{q+1}, \ldots, u_{|U(j)|}\}$. 
We show that $u$ is guarded by $u_{q}$ and $v_j$.
By the definition of $U(j)_{\leq a}$, we know that $a(u_{q}) < a(u) < a(v_j)$.
Moreover, observe that $b(u) \leq b(v_j)$ holds by the fact that $u \in V_{i,j}$ and $b(u_{q}) \leq b(v_j)$.
Thus, we apply Lemma~\ref{lem:inclusion} to $u$, because $u_{q},v_j \in C$ and $u$ is $C$-compatible, showing that $u \in C$ as desired.
\end{proof}

For $a(v_j) \leq t \leq b(v_j)$, let $M[{t}] = K_{t} \cap M(j)$.
Observe that each $M[{t}]$ may be an empty set.
On the part $M(j)$, all vertices are grouped into the sets $M[{a(v_j)}], \ldots, M[{b(v_j)}]$.
Similar to $M[{t}]$, let $U[{t}] = U(j) \cap K_{t}$.
Then, all vertices of $U[{t}]$ are $\{v_j,M[t]\}$-compatible and all vertices of $M[t]$ are $\{v_j,U[t]\}$-compatible. 
Figure~\ref{fig:CD_UM} depicts the corresponding sets. 

\begin{figure}[t]
\centering
\includegraphics[scale= 0.79]{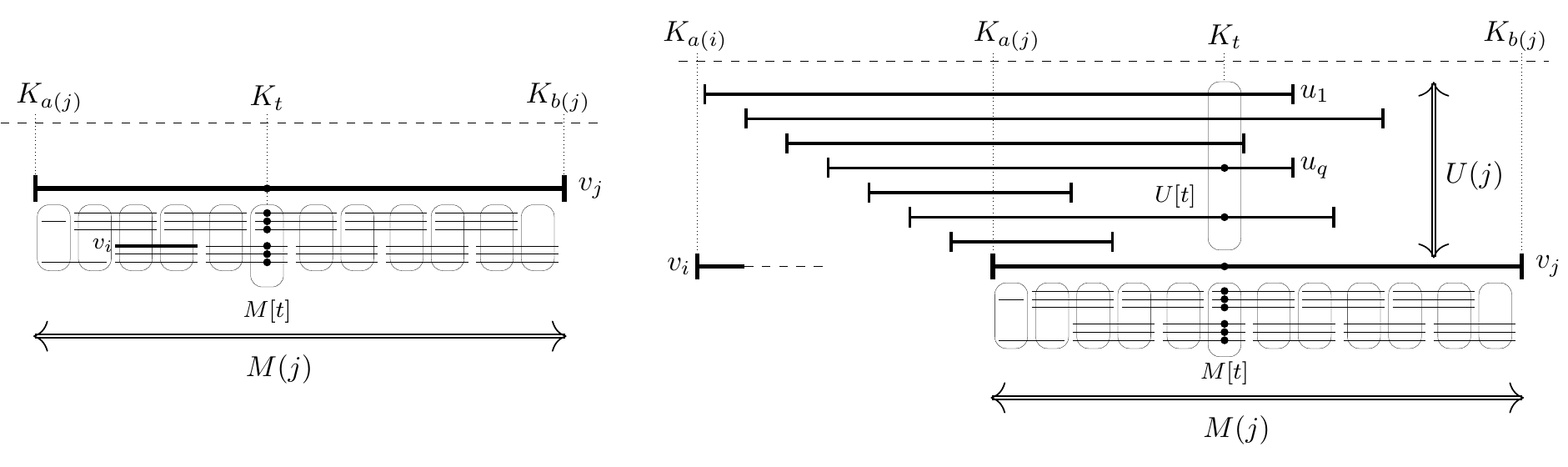}
\caption{Illustrating the sets $M(j)$ and $U(j)$ for $v_j$. 
The left part shows the case in which $v_i\in M(j)$ (or, equivalently, $V_{i,j}=V_{j,j}$), whereas the right part corresponds to the case in which $a(v_i) < a(v_j)$.}
\label{fig:CD_UM}
\end{figure}

\begin{lemma}\label{lem:existsMi1}
Let $C$ be a cluster of $A_{i,j}$ containing $v_j$.
Then, there is $a(v_j) \leq t \leq b(v_j)$ such that $M[{t}] \subseteq C$.
\end{lemma}
\begin{proof}
Assume for contradiction that no set $M[{t}]$ is contained in $C$.
Let $U_C = U(j) \cap C$ and let $i' = \bmin(U_C)$.
Notice that $C = \{v_j\} \cup U_C$ because of the assumption as there are no other neighbors of $v_j$ in $V_{i,j}$.
Then, $a(v_j) \leq i' \leq b(v_j)$ holds, because $v_j \in C$.
We show that $M[i'] \subseteq C$.
Observe that $C \subseteq K_{i'}$.
If $M[i'] = \emptyset$ then clearly $M[i'] \subset C$.
Assume that $M[i'] \neq \emptyset$ and let $C'$ be a non-empty subset of $M[i']$ that forms a cluster in $A_{i,j}$.
Then, all vertices of $C$ are $C'$-compatible and all vertices of $C'$ are $C$-compatible, because $C,C' \in K_{t}$.
Thus, we reach a contradiction by Lemma~\ref{lem:compatible} to the optimality of $A_{i,j}$.
This means that there is a vertex $w \in M(j)$ that is contained in $C$ together with $v_j$.
Therefore, by Lemma~\ref{lem:lower}, there is a set $M[t] = K_{t} \cap M(j)$ that is included in $C$.
\end{proof}

All vertices of a cluster $C$ containing $v_j$ belong to $U(j) \cup M(j)$.
Thus, $C\setminus \{v_j\}$ can be partitioned into $C \cap U(j)$ and $C\cap M(j)$.
Also notice that $C \subseteq K_{t}$ for some $a(v_j) \leq t \leq b(v_j)$.
Combined with the previous lemmas, we can enumerate all such subsets $C$ of $U(j) \cup M(j)$ in polynomial-time.
In particular, we first build all candidates for $C\cap M(j)$, which are exactly the sets $M[{t}]$ by Lemma~\ref{lem:lower} and Lemma~\ref{lem:existsMi1}.
Then, for each of such candidate $M[{t}]$, we apply Lemma~\ref{lem:upper} to construct all subsets containing the last $q$ vertices of $U[{t}]_{\leq a}$. 
Thus, there are at most $n^2$ number of candidate sets from the vertices of $U(j) \cup M(j)$ that belong to the same cluster with $v_j$.

\subsection{Splitting into partial solutions}
We further partition the vertices of $M(j)$.
Given a pivot group $M[{t}]$, we consider the vertices that lie on the right part of $M[t]$.
More formally, for $a(v_j) \leq t < b(v_j)$, we define the set
$$
B_j(t) = \left(\left(K_{t+1} \cup \cdots \cup K_{b(j)}\right) \setminus K_{t} \right) \cap M(j).
$$
The reason of breaking the vertices of the part $M(j)$ into sets $B_j(t)$ is the following.
\begin{lemma}\label{lem:sepa2}
Let $C$ be a cluster of $A_{i,j}$ such that $\{v_j\} \cup M[{t}] \subseteq C$, for $a(v_j) \leq t \leq b(v_j)$.
Then, for any two vertices $x \in V_{i,j} \setminus B_j(t)$ and $y \in B_j(t)$, there is no cluster of $A_{i,j}$ that contains both of them.
\end{lemma}
\begin{proof}
First observe that $y \in (M[{t+1}] \cup \cdots \cup M[{b(j)}]) \setminus M[t]$.
We consider two cases for $x$, depending on whether $x\in M(j)$ or not. Assume that $x \in M(j)$.
If $x \in M[{t}]$, then $x \in C$ by Lemma~\ref{lem:lower}, which implies that $y\notin C$.
If $x \in (M[{a(v_j)}] \cup \cdots \cup M[{t-1}]) \setminus M[{t}]$ then $xy \notin E(G)$.

Now assume that $x \in U(j)$.
If $x \in C$, then $y$ does not belong to $K_{t}$, so that $y \notin C$.
If $x \notin C$, then we show that $x$ does not belong to a cluster with any vertex of $B_j(t)$.
Assume for contradiction that $x$ belongs to a cluster $C'$ such that $C' \cap B_j(t) \neq \emptyset$.
This means that $x\in K_{i'}$ with $t < i' \leq b(v_j)$ and $C' \subseteq K_{i'}$.
Then $v_j$ is $C'$-compatible and $x$ is $C$-compatible, as both $x$ and $v_j$ belong to $K_{t} \cap K_{i'}$.
Therefore, by Lemma~\ref{lem:compatible} we reach a contradiction to $x$ and $v_j$ belonging to different clusters.
\end{proof}


For a non-empty set $S \subseteq V(G)$, we write $A(S)$ to denote the following solutions:
\begin{itemize}
\item $A(S) = A_{i',j'}$, where $v_{i'}$ is the vertex of $S$ having the smallest $a(v_{i'})$ and $v_{j'}$ is the vertex of $S$ having the largest $b(v_{j'})$.
\end{itemize}
Having this notation, observe that $A_{i,j} = A(V_{i,j})$, for any $v_i,v_j$ with $b(v_i)\leq b(v_j)$.
However, it is important to notice that $A(S)$ does not necessarily represent the optimal solution of $G[S]$, since the vertices of $S$ may not be consecutive with respect to $V_{i',j'}$, so that $S$ is only a subset of $V_{i',j'}$ in the corresponding solution $A_{i',j'}$ for $A(S)$.
Under the following assumptions, with the next result we show that for the chosen sets we have $S=V_{i',j'}$.

\begin{observation}\label{obs:subsetS}
Let $v_i,v_j$ be two vertices with $b(v_i) \leq b(v_j)$ and let $V_{t} = K_{t} \cap V_{i,j}$,
for any maximal clique $K_t$ of $\mathcal{P}$ with $a(v_j) \leq t \leq b(v_j)$. 
\begin{enumerate}
\item[(i)] If $S_L = \left(V_{a(i,j)} \cup \cdots \cup V_{t-1}\right) \setminus V_t$ then $S_L=V_{i',j'}$, where $i' = \amin(S_L)$ and $j'=\bmax(S_L)$. 
\item[(ii)] If $S_R = \left(V_{t+1} \cup \cdots \cup V_{b(v_j)}\right) \setminus V_t$ then $S_R=V_{i',j'}$, where $i' = \amin(S_R)$ and $j'=\bmax(S_R)$.
\end{enumerate}
\end{observation}
\begin{proof}
We prove the case for $S_L = \left(V_{a(i,j)} \cup \cdots \cup V_{t-1}\right) \setminus V_t$. 
As each $V_{t}$ contains vertices of $V_{i,j}$, we have $V_{i',j'} \subseteq V_{i,j}$.
Observe that either $a(v_{i'}) < a(v_{j'})$ or $a(v_{j'}) \leq a(v_{i'})$.
In both cases we show that $b(v_{j'})=t-1$.
Assume that there is a vertex $w \in S_L$ with $t-1 < b(w)$. Then $a(w) \leq t-1$ as $w\in S_L$, and $w \in K_t$ by the consecutiveness of the clique path.
This shows that $w\notin S_L$ because $w\in V_t$. Thus, $b(v_{j'})=t-1$.
We show that $a(v_{i'})=\min\{a(v_i),a(v_j)\}$.
If there is a vertex $w$ in $S_L$ with $a(w) < \min\{a(v_i),a(v_j)\}$ then $w \notin V_{i,j}$ leading to a contradiction that $V_{i',j'} \subseteq V_{i,j}$.
Hence we have $a(v_{i'})=\min\{a(v_i),a(v_j)\}$ and $b(v_{j'})=t-1$.
Moreover, observe that by the definition of $S_L$, we already know that $S_L \subseteq V_{i',j'}$.
Now it remains to notice that for every vertex $w$ with $\min\{a(v_i),a(v_j)\} \leq a(w)$ and $b(w) \leq t-1$ we have $w \in S_L$.
This follows from the fact that $w \in V_{a(w)} \cup \cdots \cup V_{b(w)}$ and $w \notin V_t$.
Therefore we get $S_L = V_{i',j'}$. 
Completely symmetric arguments along the previous lines, shows the case for $S_R$.
\end{proof}

Given the clique path $\mathcal{P}=K_1 \cdots K_p$, a {\it clique-index} $t$ is an integer $1 \leq t \leq p$.  
Let $\ell(j),r(j)$ be two clique-indices such that $a(i,j) \leq \ell(j) \leq a(v_j)$ and $a(v_j) \leq r(j) \leq b(v_j)$.
We denote by $\ell_r(j)$ the minimum value of $a(v)$ among all vertices of $v \in K_{r(j)} \cap V_{i,j}$ having $\ell(j) \leq a(v)$.
Clearly, $\ell(j) \leq \ell_r(j) \leq r(j)$ holds. 
A pair of clique-indices $(\ell(j),r(j))$ is called {\it admissible pair} for a vertex $v_j$,
if both $a(i,j) \leq \ell(j) \leq a(v_j)$ and $a(v_j) \leq r(j) \leq b(v_j)$ hold.
Given an admissible pair $(\ell(j),r(j))$, we define the following set of vertices:
\begin{itemize}
\item $C(\ell(j),r(j)) = \{z \in V_{i,j}: \ell_r(j) \leq a(z) \text{ and } r(j) \leq b(z)\}$.
\end{itemize}

Observe that all vertices of $C(\ell(j),r(j))$ induce a clique in $G$, because $C(\ell(j),r(j)) \subseteq K_{r(j)}$.
We say that a vertex $u$ {\it crosses} the pair $(\ell(j),r(j))$ if $a(u) < \ell_r(j)$ and $r(j) \leq b(u)$.
It is not difficult to see that for a vertex $u$ that crosses $(\ell(j),r(j))$, we have $u \notin C(\ell(j),r(j))$.
We prove the following properties of $C(\ell(j),r(j))$.

\begin{lemma}\label{lem:ellrpair}
Let $v_{i'},v_{j'}$ be two vertices with $b(v_{i'})\leq b(v_{j'})$ and let $(\ell,r)$ be an admissible pair for $v_{j'}$.
Moreover, let $v_i,v_j$ be the vertices of $V_{i',j'} \setminus C(\ell,r)$ having the smallest $a(v_i)$ and largest $b(v_j)$, respectively.
If the vertices of $C(\ell,r)$ form a cluster in $A_{i',j'}$ then 
the following statements hold:
\begin{enumerate}
\item $V_{i,j} = V_{i',j'} \setminus C(\ell,r)$.
\item If $a(x) \leq r \leq b(x)$ holds for a vertex $x \in V_{i,j}$, then $x$ crosses $(\ell,r)$. 
\item Every vertex of $B_j(r)$ does not belong to the same cluster with any vertex of $V_{i,j} \setminus B_j(r)$.
\item Every vertex that crosses $(\ell,r)$ does not belong to the same cluster with any vertex $y \in V_{i,j}$ having $\ell_{r} \leq a(y)$.
\end{enumerate}
\end{lemma}
\begin{proof} 
First we show that $V_{i,j}= V_{i',j'} \setminus C(\ell,r)$.
Assume that there is a vertex $v \in V_{i,j} \setminus V_{i',j'}$.
Then $v \notin C(\ell,r)$ and $v$ is distinct from $v_i,v_j$ because, by definition, $v_i,v_j \in V_{i',j'}$.
Also notice that $v\in V_{i,j}$ implies $a(i,j) \leq a(v)$ and $b(v) \leq b(v_j)$.
By the second inequality, we get $b(v) \leq b(v_j) \leq b(v_{j'})$.
Suppose that $a(v) < a(i',j')$.
As we already know that $a(i,j) \leq a(v)$, we conclude that $a(i,j) < a(i',j')$ leading to a contradiction that $v_i,v_j \in V_{i',j'}$.
Thus we have $a(i',j') \leq a(v)$ and $b(v) \leq b(v_{j'})$, showing that $v \in V_{i',j'}$.
This means that $V_{i,j} \subset V_{i',j'}$, so that $V_{i,j} = V_{i',j'} \setminus C(\ell,r)$.

For the second statement, observe that if $\ell_{r} \leq a(x)$ then $x \in C(\ell,r)$.
Since $x \in V_{i,j}$, we conclude that $x \notin C(\ell,r)$ by the first statement.
Thus $a(x) < \ell_{r}$ holds, implying that $x$ crosses $(\ell,r)$.

With respect to the third statement, observe that no vertex of $B_j(r)$ belongs to the clique $K_{r}$.
This means that all vertices of $B_j(r)$ belong to both sets $V_{i,j}$ and $V_{i',j'}$.
Thus Lemma~\ref{lem:sepa2} and the first statement show that no two vertices $x\in V_{i,j} \setminus B_j(r)$ and $y \in  B_j(r)$ belong to the same cluster.

For the fourth statement, let $x$ be a vertex that crosses $(\ell,r)$. By the first statement we know that $x \in V_{i,j}$.
If $r < a(y)$ then $y\in B_j(r)$ and the third statement show that $x$ and $y$ do not belong to the same cluster.
Suppose that $\ell_{r} \leq a(y) \leq r$. If $r \leq b(y)$ then $y \in C(\ell,r)$ contradicting the fact that $y\in V_{i,j}$.
Putting together, we have $\ell_{r} \leq a(y) \leq b(y) < r$.
Now assume for contradiction that $x$ and $y$ belong to the same cluster $C_{xy}$.
By the fact that $a(x) < a(y)$, observe that $a(y) \leq \amin(C_{xy}) \leq \bmin(C_{xy}) \leq \min\{b(v_j),b(y)\}$.
We consider the graph induced by $V_{i',j'}$.
We show that there is a vertex of $C_{xy}$ that is $C(\ell,r)$-compatible and there is a vertex of $C(\ell,r)$ that is $C_{xy}$-compatible.
Notice that $x$ is $C(\ell,r)$-compatible, because $x$ crosses $(\ell,r)$ so that $x \in K_{r}$.
To see that there is a vertex of $C(\ell,r)$ that is $C_{xy}$-compatible,
choose $z$ to be the vertex of $C(\ell,r)$ having the smallest $a(z)$. 
This means that $a(z) = \ell_{r}$.
Then $z$ is adjacent to every vertex of $C_{xy}$ because $a(z) \leq a(y)$ and $b(y) < r \leq b(z)$.
Thus, $z \in C(\ell,r)$ is $C_{xy}$-compatible.
Therefore, Lemma~\ref{lem:compatible} shows the desired contradiction, implying that $x$ and $y$ do not belong to the same cluster.
\end{proof}

Notice that the number of admissible pairs $(\ell(j),r(j))$ for $v_j$ is polynomial because there are at most $n$ choices for each clique-index. 
Moreover, if $v_i\in M(j)$ then $\ell(j)=a(v_j)$. 
A pair of clique-indices $(\ell,r)$ with $\ell \leq r$ is called {\it bounding pair for $v_j$} if 
either $b(v_j) < r$ holds, or $v_j$ crosses $(\ell,r)$. 
Given an bounding pair $(\ell,r)$ for $v_j$, 
we write $(\ell(j),r(j)) < (\ell,r)$ to denote the set of bounding pairs $(\ell(j),r(j))$ for $v_j$ such that 
\begin{itemize}
\item $r(j) \leq b(v_j)$, whenever $b(v_j) < r$ holds, and
\item $r(j) < \ell$, otherwise.  
\end{itemize}
Observe that if $b(v_j)<r$ holds, then $(\ell(j),r(j)) < (\ell,r)$ describes all bounding pairs for $v_j$ with no restriction, regardless of $\ell$. 
On the other hand, if $\ell < a(v_j)$ and $r \leq b(v_j)$ hold, then $(\ell,r)$ is not a bounding pair for $v_j$.  
In fact, we will show that the latter case will not be considered in our partial subsolutions. 
For any admissible pair $(\ell(j),r(j))$ and any bounding pair $(\ell,r)$ for $v_j$, observe that $v_j \in C(\ell(j),r(j))$ and $v_j \notin C(\ell,r)$.
Intuitively, an admissible pair $(\ell(j),r(j))$ corresponds to the cluster containing $v_j$, whereas a bounding pair $(\ell,r)$ forbids $v_j$ to select certain vertices as they 
have already formed a cluster that does not contain $v_j$. 

Our task is to construct subsolutions over all admissible pairs for $v_j$ with the property that the vertices of $C(\ell(j),r(j))$ form a cluster. 
To do so, we consider a vertex $v_{j'}$ with $b(v_j) \leq b(v_{j'})$ and a cluster containing $v_{j'}$. 
Let $(\ell,r)$ be an admissible pair for $v_{j'}$ such that $a(v_j) \leq r \leq b(v_j)$.  
The previous results suggest to consider solutions in which the vertices of $C(\ell,r)$ form a cluster in an optimal solution. 
It is clear that if $\ell \leq a(v_j)$ then $v_j \in C(\ell,r)$. 
Moreover, if $b(v_j) < r$, then no vertex of $V_{i,j}$ belongs to $C(\ell,r)$.
Thus, we need to construct solutions for $A_{i,j}$, whenever $(\ell,r)$ is a bounding pair for $v_j$ 
and the vertices of $C(\ell,r)$ form a cluster. 
Such an idea is formally described in the following restricted solutions. 

Let $(\ell,r)$ be a bounding pair for $v_j$. We call the following solution, $(\ell,r)$-restricted solution:
\begin{itemize}
\item $A_{i,j}[\ell,r]$ is the value of an optimal solution for {\sc Cluster Deletion} of the graph $G[V_{i,j}] - \left(C(\ell,r) \cup B_j(r)\right)$ 
such that the vertices of $C(\ell,r)$ form a cluster.
\end{itemize}
Hereafter, we assume that $B_{j}(t)$ with $t \geq b(v_j)$ corresponds to an empty set. 
Figure~\ref{fig:recursion} illustrates a partition of the vertices with respect to $A_{i,j}[\ell,r]$. 
Notice that an optimal solution $A_{i,j}$ without any restriction is described in terms of $A_{i,j}[\ell,r]$ by $A_{i,j}[1,b(v_j)+1]$, since no vertex of $V_{i,j}$ belongs to $C(1,b(v_j)+1)$.
Therefore, $A_{1,n}[1,n+1]$ corresponds to the optimal solution of the whole graph $G$. 
As base cases, observe that if $V_{i,j}$ contains at most one vertex then $A_{i,j}[\ell,r]=0$ for all bounding pairs $(\ell,r)$, since there are no internal edges. 
For a set $C$, we write $|C|_2$ to denote the number ${{|C|}\choose 2}$. 
With the following result, we describe a recursive formulation for the optimal solution $A_{i,j}[\ell,r]$, which is our central tool for our dynamic programming algorithm. 

\begin{figure}[t]
\centering
\includegraphics[scale= 1.0]{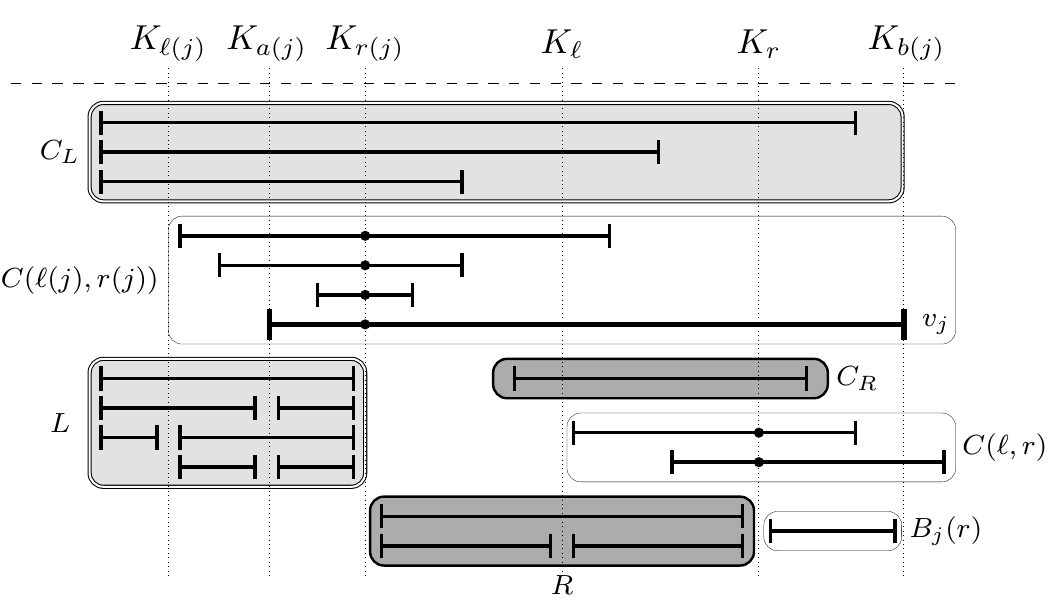}
\caption{A partition of the set of vertices given in $A_{i,j}[\ell,r]$, where $V_L = C_L \cup L$ and $V_R = C_R \cup R$. Observe that $B_{j}(r(j)) = R \cup C_R \cup \left(C\left(\ell,r\right) \cap V_{i,j}\right) \cup B_{j}(r)$.}
\label{fig:recursion}
\end{figure}

\begin{lemma}\label{lem:recurrence}
Let $(\ell,r)$ be a bounding pair for $v_j$. Then, 
$$
A_{i,j}[\ell,r] = \max_{(\ell(j),r(j))<(\ell,r)}
\left( A(V_L)[\ell(j),r(j)] + 
|C(\ell(j),r(j))|_2
+ A(V_R)[\ell,r] \right),
$$
where $V_L = V_{i,j} \setminus \left(C(\ell(j),r(j)) \cup B_j(r(j))\right)$ and 
$V_R = B_j(r(j)) \setminus \left(C(\ell,r) \cup B_{j}(r)\right)$. 
\end{lemma}
\begin{proof}
We first argue that $C(\ell(j),r(j))$ corresponds to the correct cluster $C$ containing $v_j$. 
Observe that $v_j \notin C(\ell,r)$, because $(\ell,r)$ is a bounding pair for $v_j$, so that $a(v_j) < \ell$ whenever $a(v_j) \leq r \leq b(v_j)$ holds.
By Lemmas~\ref{lem:upper} and \ref{lem:lower}, there are $r(j)=t$ and $\ell(j)=k$, where $a(v_j) \leq t \leq b(v_j)$ and $k=\amin(K_{t} \cap C)$, such that $C = C(\ell(j),r(j))$.
We show that such a set $C(\ell(j),r(j))$ is obtained from a correct choice among the described $(\ell(j),r(j))$. 
Assume first that $b(v_j) <r$.  
Then $A_{i,j}[\ell,r]=A_{i,j}$, because for every vertex $u$ of $C(\ell,r)$ we know the $b(v_j) < b(u)$, so that $V_{i,j} \cap C(\ell,r) = \emptyset$. 
This means that $a(v_j) \leq r(j) \leq b(v_j)$ for every bounding pair $(\ell(j),r(j))$, as described in the given formula. 
Now assume that $r \leq b(v_j)$. 
Since $v_j$ crosses $(\ell,r)$, Lemma~\ref{lem:ellrpair}~(4) shows that $v_j$ is not contained in a cluster with a vertex $y$ having $\ell < a(y)$. 
Thus, for any vertex $y \in C$ we know that $y \in K_t$ where $a(v_j) \leq t < \ell$. 
This means that there is a set $C(\ell(j),r(j))$ that contains exactly the vertices of $C$ such that $a(v_j) \leq r(j) < \ell$. 
Therefore, $(\ell(j),r(j)) < (\ell,r)$ holds, as desired.  

Next, we consider the sets $V_L$ and $V_R$. 
We show that $A(V_L)[\ell(j),r(j)]$ and $A(V_R)[\ell,r]$ correctly store the optimal values of each part. 
To do so, we show first that the vertex sets of each part correspond to the correct sets and, then, each pair $(\ell(j),r(j))$ and $(\ell,r)$ is indeed a bounding pair for the last vertex of $V_L$ and $V_R$, respectively.
We start with some preliminary observations. 
Notice that $B_{j}(r) \subseteq B_{j}(r(j))$, because $r(j) < r$, which means that every vertex $B_{j}(r)$ does not belong to $V_L \cup V_R$. 
Since $C(\ell(j),r(j))$ contains only vertices of $K_{r(j)}$ and $r(j)<\ell$, no vertex of $B_{j}(r)$ is considered in the described formula, as required in $A_{i,j}[\ell,r]$.
By the properties of $C(\ell(j),r(j))$ and $C(\ell,r)$, we have the following: 
\begin{itemize}
\item Let $x \in K_{r(j)} \cap V_{i,j}$. 
Then, either $x \in C(\ell(j),r(j))$ or $x$ crosses the pair $(\ell(j),r(j))$. 
Moreover, if a vertex $v$ crosses $(\ell(j),r(j))$ then $v \in V_L$. 
\item Let $y \in K_{r} \cap V_{i,j}$. 
Then, either $y \in C(\ell,r)$ or $y$ crosses the pair $(\ell,r)$. 
Moreover, if a vertex $v$ crosses $(\ell,r)$ but does not cross $(\ell(j),r(j))$ then $v \in V_R$. 
\end{itemize}
Let $C_L$ be the set of vertices of $V_{i,j}$ that cross $(\ell(j),r(j))$ and let $C_R$ be the set of vertices of $V_{i,j} \setminus C_L$ that cross $(\ell,r)$. 
The previous properties imply that we can partition $V_L$ to the vertices of $C_L$ and the vertices of $V_{i,j}$ that belong to $L=(K_{a(i,j)} \cup \cdots \cup K_{r(j)-1}) \setminus K_{r(j)}$. 
Similarly, $V_R$ is partitioned to the vertices of $C_R$ and the vertices of $V_{i,j}$ that belong to $R=(K_{r(j)+1} \cup \cdots \cup K_{r-1}) \setminus (K_{r(j)} \cup K_{r})$. 
See Figure~\ref{fig:recursion} for an exposition of the corresponding sets. 
Thus, we have the following partitions for $V_L$ and $V_R$: 
\begin{itemize}
\item $V_L = C_L \cup L$, where $L = \left((K_{a(i,j)} \cup \cdots \cup K_{r(j)-1}) \setminus K_{r(j)}\right) \cap V_{i,j}$.
\item $V_R = C_R \cup R$, where $R = \left((K_{r(j)+1} \cup \cdots \cup K_{r-1}) \setminus (K_{r(j)} \cup K_{r})\right)\cap V_{i,j}$.
\end{itemize}

Let $v_{i'},v_{j'}$ be the vertices of $V_L$ with $i' = \amin(V_L)$ and $j' = \bmax(V_L)$. 
We now show that $A(V_L)[\ell(j),r(j)]$ corresponds to the optimal solution of the graph 
$G[V_{i',j'}] - \left(B_{j'}(r(j)) \cup C(\ell(j),r(j))\right)$
such that the vertices of $C(\ell(j),r(j))$ form a cluster. 
Assume for contradiction that there is a vertex $x$ of $V_{i',j'} \setminus \left(C(\ell(j),r(j)) \cup B_{j'}(r(j)) \right)$ that does not belong to 
$V_L=V_{i,j} \setminus \left(B_{j}(r) \cup C(\ell,r)\right)$. 
First notice that $K_{r(j)} \cap V_{i,j} = C(\ell(j),r(j))$ if and only if $C_L$ is an empty set. 
In such a case, by Observation~\ref{obs:subsetS}, we have $V_{i',j'} = V_{i,j} \setminus \left(K_{r(j)} \cup \cdots \cup K_{b(j)}\right)$, contradicting the existence of such a vertex $x$. 
Suppose that $v_{i'} \neq v_i$.
Then $v_{i} \in M(j)$ or $v_{i} \in C(\ell(j),r(j))$, because $\min\{a(v_i),a(v_j)\}$ is the first maximal clique of all vertices of $V_{i,j}$. 
If $v_{i} \in M(j)$ then $U(j)=\emptyset$ and $\ell(j)=a(j)$. 
This means that for every $a(v_j) \leq r(j) \leq b(v_j)$, we have $K_{r(j)} \cap V_{i,j} = C(\ell(j),r(j))$, reaching a contradiction. 
If $v_{i} \in C(\ell(j),r(j))$ then $\ell(j)=a(v_i)$ and $C_L$ is empty, reaching again a contradiction. 
Suppose now that $i'=i$. 
It is clear that $x \neq v_{j'}$. If $v_{j'} \in L$ then $C_L = \emptyset$, so that $K_{r(j)} \cap V_{i,j} = C(\ell(j),r(j))$. 
Assume that $v_{j'} \in C_L$. 
Now observe that if $x \in L \cup C_L$, then $x$ is a vertex of $V_{i,j} \setminus \left(B_{j}(r) \cup C(\ell,r)\right)$. 
Thus, $x \notin L \cup C_L$. 
If $b(x) < r(j)$ then $x \in L$ because $a(v_i) \leq a(x)$. 
This means that $r(j) \leq b(x)$. 
If $\ell(j) \leq a(x) \leq r(j)$ then $x \in C(\ell(j),r(j))$, leading to a contradiction that $x \in V_L$, 
and if $a(x) < \ell(j)$ then $x \in C_L$, leading to a contradiction that $x \notin L \cup C_L$. 
Thus, we know that $r(j) < a(x)$ and $b(x) \leq b(v_{j'})$. This, however, implies that 
$x \in B_{j'}(r(j))$, reaching a contradiction to the fact that $x \in V_{i',j'} \setminus B_{j'}(r(j))$.
Therefore, we have shown that an optimal solution of the vertices of $V_{i',j'} \setminus \left(B_{j'}(r(j)) \cup C(\ell(j),r(j))\right)$ corresponds to an optimal solution of the vertices of $V_L$. 

Furthermore, we argue that $(\ell(j),r(j))$ is a bounding pair for $v_{j'}$ in $A(V_L)[\ell(j),r(j)]$. 
Assume that $r(j) \leq b(v_{j'})$. 
If $r(j) \leq a(v_{j'})$ then $v_{j'} \in B_{j}(r(j))$, because $a(v_j) \leq r(j)$. 
As $v_{j'} \in V_L$, we have $a(v_{j'}) < r(j) \leq b(v_{j'})$. 
Then, if $\ell(j) \leq a(v_{j'})$, we get $v_{j'} \in C(\ell(j),r(j))$, which implies that $a(v_{j'})<\ell(j)$, showing that $(\ell(j),r(j))$ is a bounding pair for $v_{j'}$.  
Assume next that $b(v_{j'}) < r(j)$. Then, $v_{j'} \notin C_L$, implying that $C_L=\emptyset$. 
Thus, for any value of $\ell(j)$ we know that $(\ell(j),r(j))$ is a bounding pair for $v_{j'}$.
Therefore, $A(V_L)[\ell(j),r(j)]$ corresponds to the optimal solution of the graph
$G[V_{i',j'}] - \left(B_{j'}(r(j)) \cup C(\ell(j),r(j))\right)$. 

Next we consider the vertices of $V_R$, in order to show that $A(V_R)[\ell,r]$ corresponds to an optimal solution of the graph $G[V_R]$. 
Let $v_{i''},v_{j''}$ be the vertices of $V_R$ with $i'' = \amin(V_R)$ and $j'' = \bmax(V_R)$.
Assume for contradiction that there is a vertex $x$ of $V_{i'',j''} \setminus \left(C(\ell,r) \cup B_{j''}(r)\right)$ that does not belong to
$V_R=B_j(r(j)) \setminus \left(C(\ell,r) \cup B_{j}(r)\right)$. 
Every vertex of $R \cup C_R$ belongs to $V_R$, so that $x \notin R \cup C_R$. 
This means that $b(x) > r$, since $x \notin R$, and $a(x) > r$, since $x \notin C_R \cup C(\ell,r)$. 
Then we obtain $r < a(x) \leq b(x) \leq b(v_{j''})$, showing that $x \in B_{j''}(r)$. 
Thus we reach a contradiction, because $B_{j''}(r) \subseteq B_{j}(r)$. 
Hence, the vertices described in $A(V_R)[\ell,r]$ correspond to the vertices of $V_R$, as desired.

With respect to $A(V_R)[\ell,r]$, it remains to show that $(\ell,r)$ is a bounding pair for $v_{j''}$. 
If $b(v_{j''}) < r$ then $C_R = \emptyset$, which means that $(\ell,r)$ is a bounding pair for $v_{j''}$. 
Next suppose that $r \leq b(v_{j''})$. If $r \leq a(v_{j''})$ then $v_{j''} \in B_{j}(r)$, contradicting the fact that $v_{j''} \in V_R$. 
Thus, we know that $a(v_{j''}) < r \leq b(v_{j''})$. 
If further $\ell \leq a(v_{j''})$, then $v_{j''} \in C(\ell,r)$, contradicting $v_{j''} \in V_R$. 
Hence, we conclude that $v_{j''}$ crosses $(\ell,r)$, showing that $(\ell,r)$ is indeed a bounding pair for $v_{j''}$.

To complete the proof, observe that no vertex of $V_L$ belongs to the same cluster with a vertex of $V_R$ by Lemma~\ref{lem:ellrpair}~(3). 
Thus, the optimal solutions described by $A(V_L)[\ell(j),r(j)]$ and $A(V_R)[\ell,r]$ do not overlap in $A_{i,j}[\ell,r]$. 
Therefore, the claimed formula holds. 
\end{proof}

Now we are ready to obtain our main result, namely a polynomial-time algorithm for {\sc Cluster Deletion} on interval graphs.

\begin{theorem}\label{theo:polyinterval}
{\sc Cluster Deletion} is polynomial-time solvable on interval graphs.
\end{theorem}
\begin{proof}
We describe a dynamic programming algorithm that computes $A_{1,n}$ based on Lemma \ref{lem:recurrence}.
In a preprocessing step,
we first compute two orderings of the vertices according to their first $a(v)$ and last $b(v)$ maximal cliques.
Then we visit all vertices in ascending order with respect to $b(v_j)$ and for each such vertex $v_j$ we consider the vertices $v_i$ with $b(v_i) \leq b(v_j)$ in descending order with respect to $b(v_i)$. In such a way, we construct the sets $V_{i,j}$.
We use a table $\mathcal{T}[i,j,\ell,r]$ to store the values of each $A_{i,j}[\ell,r]$.
At the end, we output the maximum value of $\mathcal{T}[1,n,n+1,n+1]$ that corresponds to $A_{1,n}[n+1,n+1]$, as already explained. 
Regarding the running time, observe that the number of our table entries is at most $n^4$, as each table index is bounded by $n$.
Moreover, computing a single table entry requires $O(n^2)$ time, since we take the maximum of at most $(\ell,r)$ table entries.
Therefore, the overall running time of the algorithm is $O(n^6)$.
\end{proof}

\section{Cluster Deletion on a generalization of split graphs (split-twin graphs)}
A graph $G=(V,E)$ is a {\it split graph} if $V$ can be partitioned into a
clique $C$ and an independent set $I$, where $(C,I)$ is called a
{\it split partition} of $G$.
Split graphs are characterized as $(2K_2,C_4,C_5)$-free graphs \cite{FH77}. 
They form a subclass of the larger and
widely known graph class of {\it chordal graphs}, which are the graphs that do not contain
induced cycles of length $4$ or more as induced subgraphs.
In general, a split graph can have more than one split partition and computing such a partition can be done in linear time \cite{HS81}.

Hereafter, for a split graph $G$, we denote by $(C,I)$ a split partition of $G$ in which $C$ is a maximal clique.
It is known that {\sc Cluster Deletion} is polynomial-time solvable on split graphs \cite{BDM15}.
In fact, the algorithm given in \cite{BDM15} is characterized by its simplicity due to the following elegant characterization of an optimal solution:
if there is a vertex $v \in I$ such that $N(v) = C\setminus\{w\}$ and $w$ has a neighbor $v'$ in $I$ then the non-trivial clusters of an optimal solution are $C\setminus\{w\}\cup \{v\}$ and $\{w,v'\}$; otherwise, the only non-trivial cluster of an optimal solution is $C$ \cite{BDM15}.
Here we study whether such a simple characterization can be extended into more general classes of split graphs.
Due to Lemma~\ref{lem:truetwin}, it is natural to consider true twins at the independent set, as they are grouped together in an optimal solution and they are expected not to influence the solution characterization.
Surprisingly, we show that {\sc Cluster Deletion} remains NP-complete even on such a slight generalization of split graphs.
Before presenting our NP-completeness proof, let us first show that such graphs form a proper subclass of $P_5$-free chordal graphs.
We start by giving the formal definition of such graphs.

\begin{definition}\label{def:splittwin}
A graph $G=(V,E)$ is called {\it split-twin} graph if its vertex set can be partitioned into $C$ and $I$ such that $G[C]$ is a clique and the vertices of each connected component of $G[I]$ form true twins in $G$.
\end{definition}

\begin{figure}[t]
\centering
\includegraphics[scale= 1.0]{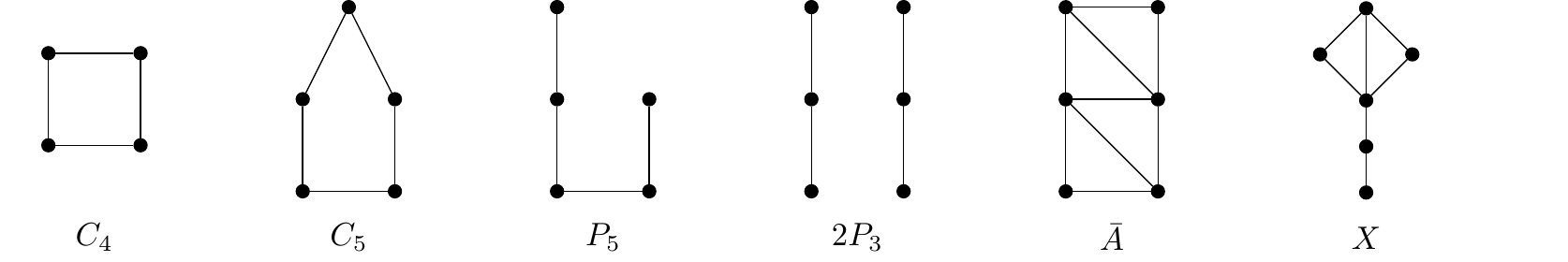}
\caption{The list of forbidden induced subgraph characterization for split-twin graphs.}\label{fig:forb}
\end{figure}

It is clear that in a split-twin graph $G$ the following holds:
(i) each connected component of $G[I]$ is a clique and forms a true-twin set in $G$, and
(ii) contracting the connected components of $G[I]$ results in a split graph, denoted by $G^*$.
Figure~\ref{fig:forb} illustrates the induced subgraphs that are forbidden in a split-twin graph.

\begin{proposition}\label{prop:forbiddensplittwin}
A graph $G$ is split-twin if and only if it does not contain any of the graphs $C_4,C_5,P_5,2P_3,\bar{A},X$ as induced subgraphs.
\end{proposition}
\begin{proof}
Let $F$ be the list of such subgraphs, i.e., $F=\{C_4,C_5,P_5,2P_3,\bar{A},X\}$. We show that split-twin graphs are exactly the $F$-free graphs.
It is clear that any subgraph of $F$ does not contain true twins.
Moreover, besides $C_4$ and $C_5$, each of the rest of the subgraphs contains an induced $2K_2$,
which implies that all such subgraphs of $F$ are not split-twin graphs.
Thus, if a graph $G$ contains one of the subgraphs of $F$ then $G$ is not a split-twin graph.

We show that any $F$-free graph $G$ is split-twin.
If $G$ is a split graph then, by definition, $G$ is split-twin.
Assume that $G$ is not a split graph.
Since $G$ does not contain $C_4$ or $C_5$ and split graphs are exactly the $(2K_2,C_4,C_5)$-free graphs, there is an induced $2K_2$ in $G$.
Let $x_1x_2$ and $y_1y_2$ be the two edges of an induced $2K_2$.
We show that the endpoints of at least one of the two edges are true twins.
Assume for contradiction that neither $x_1,x_2$ nor $y_1,y_2$ are true twins in $G$.
Let $a$ be a neighbor of $x_1$ that is non-adjacent to $x_2$, and let $b$ be a neighbor of $y_1$ that is non-adjacent to $y_2$.
We show that the vertices of $\{a,x_1,x_2,b,y_1,y_2\}$ induce one of the subgraphs of $F$, contradicting the fact that no pair of vertices form true twins.
If $b \notin N(\{x_1,x_2\})$ and $a \notin N(\{y_1,y_2\})$ then there is an induced $P_5$ or $2P_3$ depending on whether $a$ and $b$ are adjacent or not.
Thus, $b \in N(\{x_1,x_2\})$ or $a \in N(\{y_1,y_2\})$.
Observe that if $a$ is adjacent to at least one of $y_1$ or $y_2$ then $a$ is adjacent to both $y_1$ and $y_2$; otherwise, $\{x_1,x_2,a,y_1,y_2\}$ induce a $P_5$.
By symmetric arguments we know that either $b$ is adjacent to both $x_1,x_2$ or to none.
Without loss of generality, assume that $bx_1, bx_2 \in E(G)$.
\begin{itemize}
\item Suppose that $a$ and $b$ are non-adjacent.
If $a \notin N(\{y_1,y_2\})$ then there is a $P_5$ induced by $\{a, x_1, b, y_1, y_2\}$.
Moreover, by the previous argument, we know that if $a \in N(\{y_1,y_2\})$ then $ay_1, ay_2 \in E(G)$, which implies a $C_4$ in $G$ induced by $\{a,x_1,b,y_1\}$.
Thus if $ab \notin E(G)$ we obtain a induced subgraph of $F$.
\item Suppose that $a$ and $b$ are adjacent. If $a \notin N(\{y_1,y_2\})$, then all six vertices induce an $X$ graph.
Otherwise, we know that $ay_1, ay_2 \in E(G)$, showing that all six vertices induce a graph $\bar{A}$, where $a$ and $b$ are the degree four vertices.
\end{itemize}
Thus in all cases we obtain an induced subgraph of $F$, reaching to a contradiction that $G$ being an $F$-free graph.
This means that for any $2K_2$ we know that at least one of the two edges contains true twin vertices in $G$.
By iteratively picking such true twins and contracting them into a new vertex, results in a graph $G^*$ that does not contain $2K_2$.
Therefore $G^*$ is a split graph, implying that $G$ is a split-twin graph.
\end{proof}

Thus by Proposition~\ref{prop:forbiddensplittwin}, split-twin graphs form a proper subclass of $P_5$-free chordal graphs, i.e., of $(C_4,C_5,P_5)$-free graphs.
Now let us show that decision version of {\sc Cluster Deletion} is NP-complete on split-twin graphs.
For the reduction we will use the NP-hard {\sc Edge Weighted Cluster Deletion} problem.
In the {\sc Edge Weighted Cluster Deletion} problem, each edge of the input graph is associated with a weight and
the objective is to construct a clustered graph having the maximum total (cumulative) weight of edges.
It is known that {\sc Edge Weighted Cluster Deletion} remains NP-hard on split graphs even when (i) all edges inside the clique have weight one, (ii) all edges incident to a vertex $w \in I$ have the same weight $q$, and (iii) $q=|C|$ \cite{BDM15}.
We abbreviate the latter problem by EWCD and denote by $(C,I,k)$ an instance of the problem where $(C,I)$ is a split partition of the vertices of $G$ and $k$ is the total weight of the edges in a cluster solution for $G$.

\begin{theorem}\label{theo:npsplittwin}
The decision version of {\sc Cluster Deletion} is NP-hard on split-twin graphs.
\end{theorem}
\begin{proof}
We prove the NP-hardness of the {\sc Cluster Deletion} problem on split-twin graphs
by giving a polynomial reduction from restricted version EWCD of {\sc Edge Weighted Cluster Deletion} on split graphs which is known to be NP-hard \cite{BDM15}.
Let $(C,I,k)$ be an instance of EWCD, where $G=(C \cup I, E)$ is a split graph.
From $G$, we build a split-twin graph $G'=(C' \cup I', E')$ by keeping the same clique $C'=C$, and for every vertex $w_j \in I$ we apply the following:
\begin{itemize}
\item We replace $w_j$ by $q=|C|$ true twin vertices $I'_j$ (i.e., by a $q$-clique) such that for any vertex $w' \in I'_j$ we have $N_{G'}(w')= N_G(w_j) \cup (I'_j \setminus \{w'\})$.
That is, their neighbors outside $I'_j$ are exactly $N_G(w_j)$.
Moreover, the set of vertices $I'_1, \ldots, I'_{|I|}$ form $I'$.
\end{itemize}
By the above construction, it is not difficult to see that $G'$ is a split-twin graph, since the graph induced by $I'$ is a disjoint union of cliques and two adjacent vertices of $I'$ are true twins in $G'$.
Also observe that the construction takes polynomial time because $q$ is at most $n = |V(G)|$.
We claim that there is an edge weighted cluster solution for $G$ with total weight at least $k$
if and only if there is a cluster solution for $G'$ having at least $k+|I|\cdot{{q}\choose 2}$ edges.

Assume that there is a cluster solution $S$ for $G$ with total weight at least $k$.
From $S$, we construct a solution $S'$ for $G'$ having the desired number of edges.
There are three types of clusters in $S$:
\begin{itemize}
\item[(a)] Cluster formed only by vertices of the clique $C$, i.e., $Y \in S$, where $Y \subseteq C$. We keep such clusters in $S'$. We denote by $t_a$ the total weight of clusters of type (a). Notice that since the weight of edges having both endpoints in $C$ are all equal to one, $t_a$ corresponds to the number of edges in $Y$.
\item[(b)] Cluster formed only by one vertex $w_j \in I$, i.e., $\{w_j\} \in S$. In $S'$ we replace such cluster by the corresponding clique $I'_j$ having exactly ${q \choose 2}$ edges. It is clear that total weight of such clusters do not contribute to the value of $S$.
\item[(c)] Cluster formed by the vertices $y_1, \ldots, y_p,w_j$, where $y_i \in C$ and $w_j \in I$. As the weights of the edges between the vertices of $y_i$ is one, the total number of weights in such a cluster is ${p \choose 2} + p\cdot q$. Let $t_c$ be the total weight of clusters of type (c). In $S'$ we replace $w_j$ by the vertices of $I'_j$ and obtain a cluster $S'$ having ${p \choose 2} + p\cdot q + {q \choose 2}$ number of edges.
\end{itemize}
Now observe that in $S$ we have $t_a + t_c$ total weight, which implies $t_a+t_c \geq k$. Thus, in $S'$ we have at least $t_a+t_c+|I|\cdot{{q}\choose 2}$ edges, giving the desired bound.

For the opposite direction, assume that there is a cluster solution $S'$ for $G'$ having at least $k+|I|\cdot{{q}\choose 2}$ edges.
All vertices of $I'_j$ are true twins and, thus, by Lemma~\ref{lem:truetwin} we know that they belong to the same cluster in $S'$.
Thus, any cluster of $S'$ has one of the following forms:
(i) $Y'$, where $Y' \subseteq C'$,
(ii) $I'_j$,
(iii) $I'_j \cup \{y'_1, \ldots, y'_p\}$, where $y'_i \in C'$.
This means that all internal edges having both endpoints in $I'$ contribute to the value of $S'$ by $|I|\cdot{{q}\choose 2}$.
Moreover, observe that for any internal edge of $S'$ of the form $y'w'$ with $y'\in C'$ and $w' \in I'_j$, we know that there are exactly $q$ internal edges incident to $y'$ and the $q$ vertices of $I'_j$. Thus such internal edges $y'w'$ of $S'$ correspond to exactly one internal edge $yw_j$ of $S$ having weight $q$ where $y=y'$ (because $C=C'$) and $w_j$ is the vertex of $I$ associated with $I_j$.
Hence, all internal edges outside each $I'_j$ in $S'$ correspond to either a weighted internal edge in $S$ or to the same unweighted edge of the clique $C$ in $S$.
Therefore, there is an edge weighted solution $S$ having weight at least $k$.
\end{proof}

\subsection{Polynomial-time algorithms on subclasses of split-twin graphs}
Due to the hardness result given in Theorem~\ref{theo:npsplittwin}, it is natural to consider subclasses of split-twin graphs related to their analogue subclasses of split graphs.
We consider two such subclasses. One of them corresponds to the split-twin graphs such that the vertices of $I$ have no common neighbor in the clique, unless they are true or false twins.
The other one corresponds to threshold graphs (i.e., split graphs in which the vertices of the independent set have nested neighborhood) and form the split-twin graphs in which the vertices of $I$ have a nested neighborhood.
We formally define such graphs and give polynomial-time algorithms for {\sc Cluster Deletion} on both graph classes.
For a vertex $x \in I$ we write $N_{C}(x)$ to denote the set $N(x) \cap C$.

\begin{definition}
A split-twin graph $G$ with partition $(C,I)$ on its vertices is called {\it 1-split-twin} graph if for any two vertices $x,y \in I$,
either $N_C(x) \cap N_C(y) =\emptyset$ or $N_C(x) = N_C(y)$.
\end{definition}

It is not difficult to see that in a 1-split-twin graph, any two vertices of $I$ having a common neighbor in $C$ have exactly the same neighborhood in $C$.

\begin{theorem}\label{theo:onesplittwin}
{\sc Cluster Deletion} is polynomial-time solvable on 1-split-twin graphs.
\end{theorem}
\begin{proof}
Let $G$ be a 1-split-twin graph with partition $(C,I)$.
First observe that if $G$ is disconnected then $I$ contains isolated cliques, i.e., true twins having no neighbor in $C$.
Thus we can restrict ourselves to a connected graph $G$, since by Lemma~\ref{lem:truetwin} each isolated clique is contained in exactly one cluster of an optimal solution.
We now show that all vertices of $C$ that have a common neighbor in $I$ are true twins.
Let $u$ and $v$ be two vertices of $C$ such that $x \in N(u) \cap N(v) \cap I$.
All vertices of $C \setminus \{u,v\}$ are adjacent to both $u$ and $v$.
Assume that there is a vertex $y \in I$ that is adjacent to $u$ and non-adjacent to $v$.
If $xy \in E(G)$ then by the definition of split-twin graphs $x$ and $y$ are true twins which contradicts the assumption of $xv\in E(G)$ and $yv\notin E(G)$.
Otherwise, $x$ and $y$ are non-adjacent and since $N_C(x) \cap N_C(y) \neq \emptyset$ we reach a contradiction to the definition of 1-split-twin graphs.
Thus, all vertices of $C$ that have a common neighbor in $I$ are true twins.

We partition the vertices of $C$ into true twin classes $C_1, \ldots, C_k$,
such that each $C_i$ contains true twins of $C$. 
From the previous discussion, we know that any vertex of $I$ is adjacent to all the vertices of exactly one class $C_i$;
otherwise, there are vertices of different classes in $C$ that have common neighbor.
For a class $C_i$, we partition the vertices of $N(C_i) \cap I$ into true twin classes $I_i^{1}, \ldots, I_i^{q}$ such that $|I_i^{1}|\geq \cdots \geq |I_i^{q}|$.

We claim that in an optimal solution $S$, the vertices of each class $I_i^{j}$ with $j\geq 2$ constitute a cluster.
To see this, observe first that the vertices of $I_i^{j}$, $1\leq j \leq q$, are true twins, and by Lemma~\ref{lem:truetwin} they all belong to the same cluster of $S$.
Also, by Lemma~\ref{lem:truetwin} we know that all the vertices of $C_i$ belong to the same cluster of $S$.
Moreover, all vertices between different classes $I_i^{j}$,$I_i^{j'}$ are non-adjacent and are $C_i$-compatible.
Since every vertex of $I_i^{j}$ is non-adjacent to all the vertices of $V(G)\setminus \{I_i^{j} \cup C_i\}$,
we know that any cluster of $S$ that contains $I_i^{j}$ is of the form either $\{I_i^{j} \cup C_i\}$ or $I_i^{j}$.
Assume that there is a cluster that contains $\{I_i^{j} \cup C_i\}$ with $j\geq 2$.
Then, we substitute the vertices of $I_i^{j}$ by the vertices of $I_i^{1}$ and obtain a solution of at least the same size, because $|I_i^{1}|\geq |I_i^{j}|$ implies  ${{|C_i|+|I_i^{1}|}\choose 2} \geq {{|C_i|+|I_i^{j}|}\choose 2}$.
Thus, all vertices of each class $I_i^{j}$ with $j\geq 2$ constitute a cluster in an optimal solution $S$.

This means that we can safely remove the vertices of $I_i^{j}$ with $j\geq 2$, by constructing a cluster that contains only $I_i^{j}$.
Hence, we construct a graph $G^*$ from $G$, in which there are only matched pair of $k$ classes $(C_i,I_i)$ such that
(i) all sets $C_i,I_i$ are non-empty except possibly the set $I_k$,
(ii) $N(C_i)\cap I = I_i$, (iii) $N(I_i)=C_i$, (iv) $G^*[C_i \cup I_i]$ is a clique, and (v) $G^*[C_1 \cup \cdots \cup C_k]$ is a clique.
Our task is to solve {\sc Cluster Deletion} on $G^*$, since for the rest of the vertices we have determined their cluster.
By Lemma~\ref{lem:truetwin}, observe that if the vertices of $C_i \cup C_j$ belong to the same cluster then the vertices of each $I_i$ and $I_j$ constitute two respectively clusters. Thus, for each set of vertices $I_i$ we know that either one of $C_i \cup I_i$ or $I_i$ constitutes a cluster in $S$.
This boils down to compute a set $M$ of matched pairs $(C_i,I_i)$ from the $k$ classes, having the maximum value
$$
\sum_{(C_i,I_i)\in M} {|C_i|+|I_i|\choose 2} + {{\sum_{C_j \notin M} |C_j|} \choose 2} + \sum_{I_j \notin M}{|I_j| \choose 2}.
$$
Let $(C_i,I_i)$ and $(C_j,I_j)$ be two pairs of classes such that $|C_i|+|I_i| \leq |C_j|+|I_j|$.
We show that if $(C_j,I_j) \notin M$ then $(C_i,I_i) \notin M$.
Assume for contradiction that $(C_j,I_j) \notin M$ and $(C_i,I_i) \in M$.
Observe that $|I_j| < \sum_{C_t \notin M\setminus C_j} |C_t|$, because $I_j$ is $C_j$-compatible.
Similarly, we know that $\sum_{C_t \notin M\setminus C_j} |C_t| + |C_j| \leq |I_i|$.
This however, shows that $|C_j|+|I_j| < |I_i|$, contradicting the fact that $|C_i|+|I_i| \leq |C_j|+|I_j|$.
Thus $(C_j,I_j) \notin M$ implies $(C_i,I_i) \notin M$.

This means that we can consider the $k$ pair of classes $(C_i,I_i)$ in a decreasing order according to their number of vertices $|C_i|+|I_i|$.
With a simple dynamic programming algorithm, starting from the largest ordered pair $(C_1,I_1)$ we know that either $(C_1,I_1)$ belongs to $M$ or not.
In the former, we add ${|C_1|+|I_1|\choose 2}$ to the optimal value of $(C_2,I_2), \ldots, (C_k,I_k)$ and in the latter we know that no pair belongs to $M$
giving a total value of ${{\sum |C_i|} \choose 2} + \sum{|I_i| \choose 2}$.
By choosing the maximum between the two values, we construct a table of size $k$ needed for the dynamic programming.
Computing the twin classes and the partition $(C,I)$ takes linear time in the size of $G$ and sorting the pair of classes can be done $O(n)$ time, since $\sum (|C_i|+|I_i|)$ is bounded by $n$.
Thus, the total running time is $O(n+m)$, as the dynamic programming for computing $M$ requires $O(n)$ time.
Therefore, all steps can be carried out in linear time for a 1-split-twin graph $G$.
\end{proof}

\begin{definition}
A split-twin graph $G$ with partition $(C,I)$ on its vertices is called {\it threshold-twin} graph if the vertices of $I$ can be ordered $w_1, \ldots, w_{|I|}$ such that
for any $w_i,w_j \in I$ with $i <j$, we have $N_C(w_i) \subseteq N_C(w_j)$.
\end{definition}

\begin{theorem}\label{theo:onesplittwin}
{\sc Cluster Deletion} is polynomial-time solvable on threshold-twin graphs.
\end{theorem}
\begin{proof}
Let $G$ be a threshold-twin graph with partition $(C,I)$.
We show that there is no induced path on four vertices, $P_4$, in $G$.
Assume for contradiction that there is a $P_4=v_1v_2v_3v_4$ in $G$.
Since $G[C]$ is a clique and $G[I]$ is a disjoint union of cliques, at least one of $v_1,v_4$, say $v_1$, belongs to $I$.
If $v_4\in C$ then $v_2 \in I$ because $v_4v_2 \notin E(G)$, which gives a contradiction as $v_1v_2\in E(G)$ and $v_1,v_2$ are not true twins.
Otherwise, we have $v_4\in I$, so that $v_2,v_3 \in C$ because $v_1,v_2$ and $v_3,v_4$ are not true twins $G$.
The latter, results again in a contradiction because $N(v_1) \cap C \nsubseteq N(v_4) \cap C$ and $N(v_4) \cap C \nsubseteq N(v_1) \cap C$. Thus, $G$ is a $P_4$-free graph.
Therefore, by the polynomial-time algorithm for {\sc Cluster Deletion} on $P_4$-free graphs \cite{CD-cographs}, we obtain a solution for {\sc Cluster Deletion} on $G$.
\end{proof}

\section{Concluding remarks}
It is notable that our algorithm for interval graphs, heavily relies on the linear structure obtained from their clique paths. 
Such an observation, leads us to consider few open questions regarding two main directions. 
On the one hand, it seems tempting to adjust our algorithm for other vertex partitioning problems on interval graphs within a more general framework, as already have been studied for particular graph properties \cite{Bui-XuanTV13,GERBER2003719,HLNPT10,KanjKSL18,TelleP97}.
On the other hand, it is reasonable to ask whether our approach works for {\sc Cluster Deletion} on graphs admitting similar linear structure such as permutation graphs, or graphs having bounded linear related parameter.
Towards the latter direction, observe that {\sc Cluster Deletion} as a vertex partitioning problem seems to be expressible in monadic second order logic of second type with quantification over vertex sets and edge sets. Therefore, {\sc Cluster Deletion} can be solved in linear time on graphs of
bounded treewidth by using Courcelle's machinery \cite{Courcelle90}.

Although for other structural parameters it seems rather difficult to obtain a similar result,
it is still interesting to settle the complexity of {\sc Cluster Deletion} on distance hereditary graphs that admit constant clique-width \cite{GolumbicR00}.
In fact, we would like to settle the case in which from a given cograph ($P_4$-free graph) we can append degree-one vertices.
This comes in conjunction with the 1-split-twin graphs, as they can be seen as a degree-one extension of a clique.

\bibliography{CD_classes}

\end{document}